%% file: 00-arXiv.tex
\documentclass{llncs}
\usepackage{times}
\usepackage{helvet}
\usepackage{courier}
\usepackage{amssymb,amsfonts,latexsym,eurosym,bbm}
\usepackage[cmex10]{amsmath}
\usepackage{graphicx}
\usepackage{tabularx}
\usepackage{url}
\usepackage{verbatim}
\usepackage{enumerate}
\usepackage{multirow}
\usepackage[dvipsnames]{xcolor}
\usepackage{subfig}
\usepackage{makecell}
\usepackage{algorithm}
\usepackage{algorithmicx}
\usepackage[]{algpseudocode}
\usepackage{wrapfig}

\newcommand{\var}[2]{\text{\it var$_{#1}$} (#2)}
\newcommand{\sol}[1]{\text{\it sol} (#1)}
\newcommand{\players}{{\cal P}}
\newcommand{\PBS}[1]{\let\temp=\\#1\let\\=\temp}
\newcommand{\Let}[2]{\State #1 $\gets$ #2}

\begin{document}
\frontmatter          
\pagestyle{headings}  
\mainmatter              
\title{A Complete Solver for Constraint Games}
\titlerunning{A Complete Solver for Constraint Games}
%
\author{Thi-Van-Anh Nguyen \and Arnaud Lallouet}

\authorrunning{T.-V.-A. Nguyen and A. Lallouet} 

\tocauthor{}

\institute{
Universit\'e de Caen, GREYC, Campus C\^ote de Nacre,\\
Boulevard du Mar\'echal Juin, BP 5186, 14032 Caen Cedex, France\\
\email{\{thi-van-anh.nguyen,arnaud.lallouet\}@unicaen.fr}
}

\maketitle              

\begin{abstract}\footnote{This work is supported by Microsoft Research grant MRL-2011-046.}
  Game Theory studies situations in which multiple agents having conflicting objectives have to reach a collective decision.  The question of a compact representation language for agents utility function is of crucial importance since the classical representation of a $n$-players game is given by a $n$-dimensional matrix of exponential size for each player.  In this paper we use the framework of {\em Constraint Games} in which CSP are used to represent utilities.  Constraint Programming --including global constraints-- allows to easily give a compact and elegant model to many useful games.  Constraint Games come in two flavors: Constraint Satisfaction Games and Constraint Optimization Games, the first one using satisfaction to define  boolean utilities.  In addition to multimatrix games, it is also possible to model more complex games where hard constraints forbid certain situations.  In this paper we study complete search techniques and show that our solver using the compact representation of Constraint Games is faster than the classical game solver Gambit by one to two orders of magnitude.
\end{abstract}

\input{01-introduction}

\input{02-constraint-games}
\input{03-examples}

\input{04-pruning}

\input{05-nash-enumeration}
\input{06-experiments}
\input{07-conclusion}

\bibliographystyle{splncs03}
\bibliography{cpgames}

\end{document}

%% file: 01-introduction.tex
\section{Introduction}  \label{sec:introduction}
Game theory has proven to be highly successful in modeling interaction of selfish agents \cite{neumann44a,Nash1951}.  In a strategic game, each player is given a set of actions and has to choose one to perform.  A reward is given to a player by an utility function which depends on the actions taken by all players.  One of the best known solution concepts for this type of game is the pure Nash equilibrium (PNE), which occur when no player is able to improve her utility by changing her chosen action to another one.  There are many ways do define solution concepts \cite{OsborneR:Book1994} but PNE has the notable advantage of giving a deterministic decision for the players.  Indeed, PNE for games are similar to solutions for CSP: not all games own a PNE, and when available, some PNE may be more desirable than others.  They are nevertheless a basic tool to study games.

The basic representation of games is a multimatrix {\em called normal} form whose size is exponential in the number of players.  The intractability of this representation is a severe limitation to the widespead of game-based modeling.  This key issue has been adressed by several types of compact representations.  Some are based on some assumptions on the interactions between players, like graphical games \cite{DBLP:conf/uai/KearnsLS01} or action-graph games \cite{DBLP:journals/geb/JiangLB11} while other are language-based, like Boolean Games \cite{DBLP:conf/tark/HarrensteinHMW01,DBLP:conf/jelia/DunneH04,DBLP:conf/ecai/BonzonLLZ06} or Constraint Games \cite{NguyenLB:ICTAI2013}.

We focus our interest here in Constraint Games for which utilities are expressed by Constraint Satisfaction Problems (CSP) or Constraint Optimization Problems (COP).  Constraint Games provide a rich modeling language which allows a natural formulation of the players goals.  In particular, it allows to express in a compact way most classical games like congestion games \cite{Rosenthal:IJGT1973}, network games \cite{DBLP:journals/networks/BouhtouEM07}, strategic scheduling \cite{Vocking:AGT2007}, to name a few.  In addition, hard constraints \cite{Rosen:Econometrica1965} can be provided to limit the joint strategies all players may take.  This allow unrealistic equilibria to be ruled out.  Note that these constraints are global to all players, unlike the local constraints defined in \cite{BonzonLL:Synthese2012}, and they provide crucial expressivity in modeling.

Despite the high modeling interest of games to model strategic interaction, it is still difficult to find a PNE game solver.  The normal form game solver Gambit \cite{Gambit} is currently considered as state-of-the-art.  Also non-compact logic transformations have been studied \cite{DBLP:conf/csl/VosV99,DBLP:conf/asian/FooMB04}.  For Boolean Games, there are techniques limited to specific categories of games like games with acyclic interaction graph \cite{DBLP:journals/ijar/BonzonLL09} or using specific bargaining techniques \cite{DBLP:conf/atal/DunneHKW08}.  For Constraint Games, only a solver based on local search has been proposed \cite{NguyenLB:ICTAI2013}.  Not surprisingly, large games can be solved but with no guarantee of finding an equilibrium or prove the absence of equilibrium.  This is due for a large part to the high complexity of finding a PNE \cite{DBLP:journals/jair/GottlobGS05,DBLP:conf/ecai/BonzonLLZ06}.

Few other works are related to Games and Constraint Programming.  In \cite{DBLP:conf/csclp/BordeauxP04}, it has been proposed to compute a mixed equilibrium using continuous constraints.  Some other formalism try to solve a combinatorial problem by multiple agents, either with a predefined assignment of variables to agents like in DCOP \cite{Faltings2006a} or by letting the agents select dynamically their variable like in SAT-Games \cite{DBLP:conf/sat/ZhaoM04} and Adversarial CSP \cite{DBLP:conf/ecai/BrownLCF04}.

In this paper, we prove that Constraint Games are $\Sigma^p_2$-complete like Boolean Games and then we focus on the problem of finding all PNE for a Constraint Game.  Although finding one PNE is a very interesting problem in itself, finding all of them allows more freedom for choosing equilibria that fulfill some additional requirements.  For example the correctness of the computation of Pareto Nash equilibria rely on the completeness of PNE enumeration.  We present in this paper ConGa, a new correct and complete solver for Constraint Games.  This solver is based on a fast computation of equilibrium condition we call Nash consistency and a pruning algorithm for never best responses.  We demonstrate the effectiveness of our approach on publicly available benchmarks from the Gamut suite \cite{DBLP:conf/atal/NudelmanWSL04} as well as on real-life applications.

%% file: 02-constraint-games.tex
\section{Constraint Games} \label{sec:cg}
\newcommand{\proj}[1]{|_{#1}}
Let $V$ be a set of variables and $D = (D_x)_{x \in V}$ be the family of their (finite) domains.  For $W \subseteq V$, we denote by $D^W$ the set of tuples on $W$, namely $\Pi_{x \in W} D_x$.  Projection of a tuple (or a set of tuples) on a variable (or a set of variables) is denoted by $|$: for $t \in D^V$, $t \proj{x} = t_x$, $t \proj{W} = (t_x)_{x \in W}$ and for $E \subseteq D^V$, $E \proj{W} = \{ t \proj{W} ~|~ t \in E \}$.  For $W,U \subseteq V$, the join of $A \subseteq D^W$ and $B \subseteq D^U$ is $A \Join B = \{ t \in D^{W \cup U} ~|~ t \proj{W} \in A ~\wedge~ t \proj{U} \in B \}$.  When $W \cap U = \emptyset$, we denote the join of tuples $t \in D^W$ and $u \in D^U$ by $(t,u)$.
A {\em constraint} $c=(W,T)$ is a couple composed of a subset $W = \var{}{c} \subseteq V$ of variables and a relation $T = \sol{c} \subseteq D^W$ (called {\em solutions}).
A {\em Constraint Satisfaction Problem} (or CSP) is a set of constraints.  We denote by $\sol{C} = \;\Join_{c \in C} \sol{c}$ its set of solutions.  To simplify the exposition, we identify $\sol{C}$ with its cylindric extension to all variables of $V$ (i.e. with any combination of values for the variables which do not belong to $C$).

Let $\players$ be a set of $n$ players and $V$ a finite set of variables.  The set of variables is partitioned into {\em controlled} variables $V_c = \bigcup_{i \in \players} V_i$ where $V_i$ is the subset of variables controlled by Player $i$, and $V_E$ the set of {\em uncontrolled} or {\em existential} variables ($V_E = V \setminus V_c$).
\begin{definition}[Constraint Satisfaction Game] ~
  A {\em Constraint Satisfaction Game} (or CSG) is a 4-tuple $(\players,V,D,G)$ where $\players$ is a finite set of players, $V$ is a finite set of variables composed of a family of disjoint sets $(V_i)$ for each player $i \in \players$ and a set $V_E$ of {\em existential} variables disjoint of all the players variables, $D = (D_X)_{X \in V}$ is the family of their domains and $G = (G_i)_{i \in \players}$ is a family of CSP on $V$.
\end{definition}
The CSP $G_i$ is called the {\em goal} of the player $i$.  The intuition behind CSG is that, while Player $i$ can only control her own subset of variables $V_i$, her satisfaction will depend also on variables controlled by all the other players.
The intuition behind existential variables is that they are existentially quantified (but most of the time they will be functionally defined from decision variables).
A {\em strategy} for player $i$ is an assignment of the variables $V_i$ controlled by player $i$.  A {\em strategy profile} $s=(s_i)_{i \in \players}$ is the given of a strategy for each player.
\begin{definition}[Winning strategy] ~
  A strategy profile $s$ is winning for $i$ if it satisfies the goal of $i$: $s \in \sol{G_i}$.
\end{definition}
A CSG can be interpreted as a classical game with a boolean payoff function which takes value 1 when the player's CSP is satisfied and 0 when not.

We denote by $s_{-i}$ the projection of $s$ on $V_{-i} = V \setminus V_i$.  Given a strategy profile $s$, a player $i$ has a {\em beneficial deviation} if $s \not\in \sol{G_i}$ and $\exists s'_i \in D^{V_i}$ such that $(s'_i,s_{-i}) \in \sol{G_i}$.  Beneficial deviation represents the fact that a player will try to maximize her satisfaction by changing the assignment of the variables she can control if she is unsatisfied by the current assignment.  A tuple $s$ is {\em best response} for Player $i$ if this player is not able to make any beneficial deviation.  Then we define the notion of solution of a CSG by pure Nash equilibrium:
\begin{definition}[Pure Nash Equilibrium] ~
  A strategy profile $s$ is a {\em Pure Nash Equilibrium (or PNE)} of the CSG ${\cal C}$ if and only if no player has a beneficial deviation, i.e. $s$ is {\em best response} of all players. 
\end{definition}

\begin{theorem} \label{prop:sigma}
 CSG are $\Sigma^p_2$-complete.
\end{theorem}
\begin{proof}
  The proof is adapted from the $\Sigma^p_2$-completeness of boolean games \cite{DBLP:conf/ecai/BonzonLLZ06}.
  
  Membership comes from the simple algorithm in which one guesses a strategy profile and checks that no player has a beneficial deviation.  Each verification consists in proving that a player $i$ has no solution if the strategy profile is not winning for $i$.  This verification is in coNP because for a strategy profile $s$, proving that there exists a solution for Player $i$ amounts to solve the CSP $G_i$, which is in NP.  Since the number of players is finite, there is a polynomial number of calls to a coNP oracle (actually one for each player) and thus the problem is in $\Sigma^p_2$.
  
  For hardness, we introduce a special case of CSG: the 2-players 0-sum game.  In this kind of game, when one player wins, the other player looses.  Thus it is enough to represent only the goal of the first player, the other one being deduced by negation.  Such a CSG can be represented by $(\players=\{1,2\},V,D,C)$ where $C$ is the goal of player $1$ (the goal of Player $2$ is straightforwardly deduced by negation).
  
  We perform a reduction from a $\exists\forall$-QCSP to a 2-players 0-sum CSG.  $\exists\forall$-QCSP are known to be $\Sigma^p_2$-complete.  This reduction proves that even 2-players 0-sum CSG are at least as hard than solving a $\exists\forall$-QCSP.  Together with membership of the $\Sigma^p_2$ class, it gives the exact complexity for $n$-players CSG.

  The reduction is from the QCSP $Q = \exists X \forall Y C$ where $X$ and $Y$ are disjoint sets of variables to the 2-players 0-sum CSG $G = (\{1,2\},X \cup Y \cup \{x,y\},(D,D_x,D_y),C \vee (x=y))$ where $x$ is a new variable controlled by player 1, $y$ a new variable controlled by player 2 and $D_x=D_y$ are domains composed at least of 2 elements.  It is obvious that the conversion can be performed in polynomial time.  If $Q$ is valid, then let $s_1$ be the assignment of variables of $X$ and let $s_2$ be an assignment of variables of $Y$.  Because $Q$ is valid, $\forall s'_2 \in D^Y$, $(s_1,s'_2) \in sol(C)$.  Thus $(s_1,s_2)$ is a PNE because player 1 is winning and player 2 has no beneficial deviation.  Conversely, if $Q$ is not valid then for any assignment $s_1 \in D^X$ of player 1, player 2 can play $s'_2 \in D^Y$ such that $(s_1,s'_2) \not\in sol(C)$.  Then if player 1 plays $x=v$ and if $(s_1,s_2) \in sol(C)$, then player 2 can play $s'_2$ and  $y=w$ with $w \neq v$.  Thus player 2 has a beneficial deviation and $(s_1,s_2,v,w)$ is not an equilibrium.  If $(s_1,s_2) \not\in sol(C)$ and player 2 plays $y=w$, then player 1 can play $x=w$ and player 1 has a beneficial deviation.  Thus $(s_1,s_2,w,w)$ is not an equilibrium.  In conclusion, $G$ has a PNE if and only if $Q$ is valid, proving the $\Sigma^p_2$-hardness. $\Box$
\end{proof}

The players goals could be considered as soft constraints or preferences. It may happen however some games have rules that forbid some strategy profiles as they model impossible situations.  It is natural to reject such profiles by setting {\em hard constraints} shared by all players \cite{Rosen:Econometrica1965}.  
Hard constraints can be easily expressed in the framework of Constraint Games by adding an additional CSP on the whole set of variables in order to constrain the set of possible strategy profiles:
\begin{definition}[CSG with Hard Constraints] ~
  A {\em Constraint Satisfaction Game with Hard Constraints} (or CSG-HC) is a 5-tuple $(\players,V,D,C,G)$ where $(\players,V,D,G)$ is a CSG and $C$ is a CSP on $V$.
\end{definition}
It is useful to distinguish a strategy profile which does not satisfy any player's goal from a strategy profile which does not satisfy the hard constraints.  The former can be a PNE if no player has a beneficial deviation while the latter cannot.  Therefore hard constraints provide an increase of modelling expressibility (without however changing the general complexity of CSG).

By adding an optimization condition it is possible to represent classical games.  A {\em Constraint Optimization Game} (or COG) is an extension of CSG in which each player tries to optimize her goal.  This is achieved by adding for each player $i$ an optimization condition $\min(x)$ or $\max(x)$ where $x \in V$ is a variable to be optimized by Player $i$.
\begin{definition}[Constraint Optimization Game] ~~
  A {\em Constraint Optimization Game} (or COG) is a 5-tuple $(\players,V,D,G,opt)$ where $(\players,V,D,G)$ is a CSG and $opt = (opt_i)_{i \in \players}$ is a family of optimization conditions for each player of the form $\min(x)$ or $\max(x)$ where $x \in V$.
\end{definition}
A winning strategy for player $i$ is still a strategy profile which satisfies $G_i$.  However, the notion of beneficial deviation needs to be slightly adapted.  We denote by $<_{opt_i}$ the (partial) order on strategy profiles such that $s <_{opt_i} s'$ if $s_{-i} = s'_{-i}$ and $s \proj{x} < s' \proj{x}$ when $opt_i = \min(x)$ (resp. $s \proj{x} > s' \proj{x}$ when $opt_i = \max(x)$).  Given a strategy profile $s$, a player $i$ has a {\em beneficial deviation} if $\exists s'_i \in D^{V_i}$ such that $s' = (s'_i,s_{-i}) \in \sol{G_i}$ and $s' <_{opt_i} s$.  Given this, the notion of solution is the same as for CSG.  In addition, COG can be extended with hard constraints the same way CSG are, yielding COG-HC.

%% file: 03-examples.tex
\section{Modeling with Constraint Games}  \label{sec:examples}
In this section, we show that complex games can be easily expressed using constraint games.

\begin{example}[Location Game] \label{ex:location}
In this example inspired by \cite{Hotelling:EJ1929}, $n$ ice cream vendors from a set $\players = \{1,2,..,n\}$ want to choose a location numbered from $1$ to $m$ for their stand in a street.  Each seller $i$ wants to find a location $l_i$.  She already has fixed the price of her ice cream to $p_i$ and we assume there is a customer's house at each location.  No two vendors may choose the same location.  The customers choose their vendor by minimizing the sum of the distance between their house and the seller plus the price of the ice cream.

A possible situation for $3$ sellers and 14 customers is depicted in Figure \ref{fig:location}.  The strategy profile depicted at the top of the figure is not an equilibrium since the left player can deviate and ``steal'' a customer to the middle player by shifting two positions on the right. 

\begin{figure} 
 \centering
 \includegraphics[height=3.5cm]{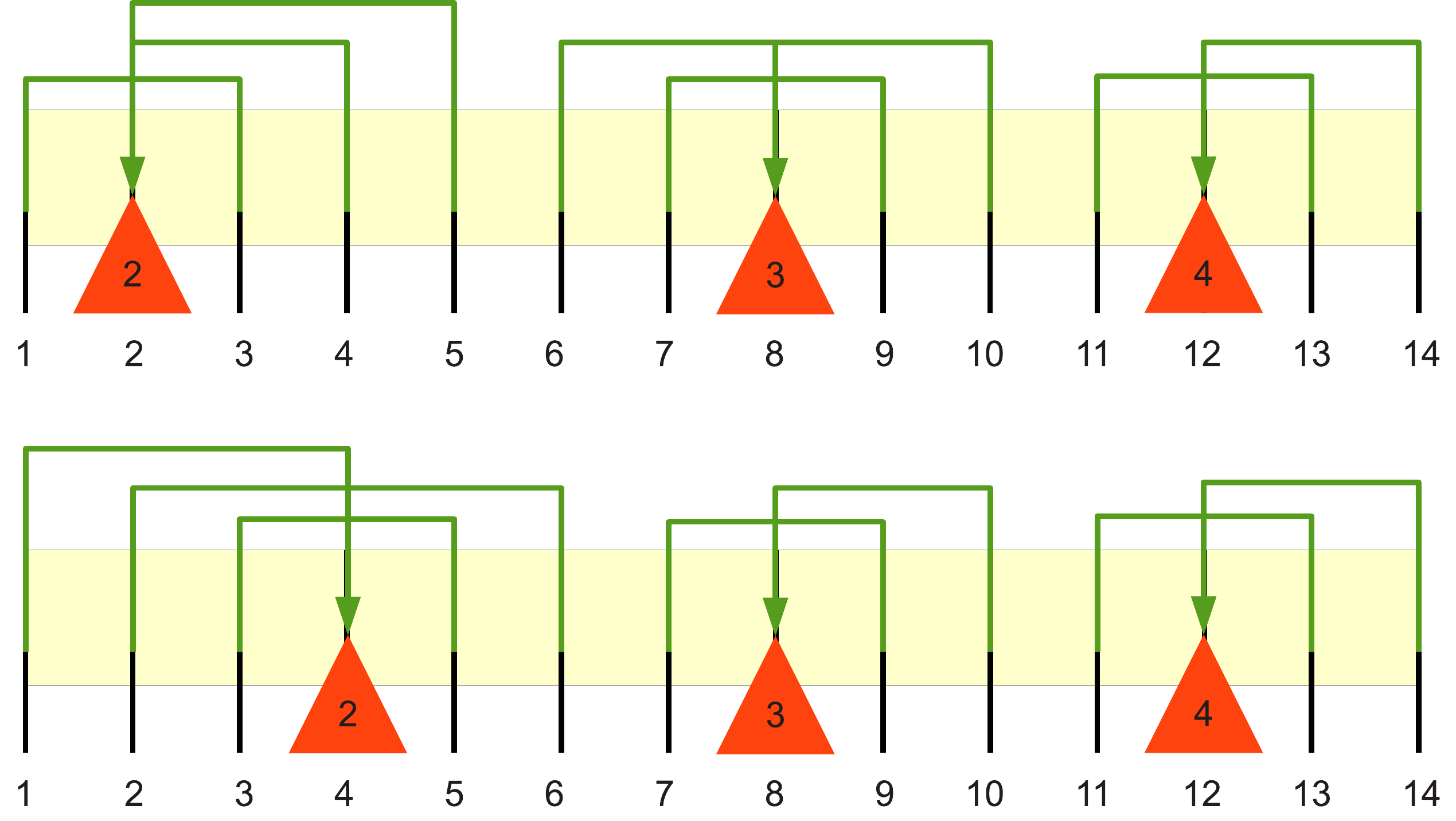}
 \caption{Location game.  Arrows depict the behavior of customers, triangles represent ice cream vendors with their selling price.  In the first situation on top, Player $1$ can improve benefits by shifting two places to the right, giving situation 2 at bottom.}
 \label{fig:location}
\end{figure} 

In order to build the model, we need the following existential variables (which are functionally determined by the decision variables $l_i$):
\begin{itemize}
  \item cost$_{ic}$: defines the cost customer $c$ has to pay if she chooses the stand of seller $i$.
  \item min$_c$: defines the minimal cost customer $c$ has to pay for an ice cream.
  \item choice$_{ic}$: boolean variable which is 1 if customer $c$ chooses seller $i$.
  \item benefit$_i$: defines the number of customers actually buying from seller $i$.
\end{itemize}
The Location Game (LG) can be easily modelled by a COG-HC in which each seller wants to maximize her profit:
\begin{itemize}
  \item $\players =\{1,\dots, n\}$
  \item $\forall i \in \players, V_i = \{l_i\}$
  \item $\forall i \in \players, D(l_i) = \{1, \ldots, m\}$
  \item the hard constraints $C$ are the following:
    \begin{itemize}
      \item no two vendors are located at the same place: {\em all\_different$(l_1,l_2,\ldots{},l_n$)}
      \item $\forall i \in \players, \forall c \in [1..m]$, cost$_{ic} = |c - l_i| + p_i$
      \item $\forall c \in [1..m]$, min$_c = \min($cost$_{1c}, \ldots, $cost$_{nc})$
      \item $\forall c \in [1..m], ($min$_c = $cost$_{ic}) \leftarrow ($choice$_{ic} = 1)$.  Because it may happen that a customer gets the same price from two sellers, we simply enforce one side implication.  Then the sum constraint on choices ensures that a customer only visits one seller.
      \item $\forall c \in [1..m], \sum_{i \in \players} $choice$_{ic} = 1$
  \end{itemize}
  \item $\forall i \in \players$, $G_i$ contains the following constraint: benefit$_i = p_i.\sum_{c \in [1..m]} $choice$_{ic}$
  \item $\forall i \in \players$, the optimization condition $Opt_i = \max($benefit$_i)$
\end{itemize}
An interesting feature of this example is that it uses global constraints like {\em all\_different} the same way as in Constraint Programming.  It also show the interest of modeling hard constraints in games since it is perfectly natural to think that no two vendors can settle at the same place.  It is possible to transform this problem into a CSG by fixing a minimal profit $mp_i$ for each player $i$ and stating that player $i$ is satisfied if her benefits is over $mp_i$.  It can be done by adding the constraint benefit$_i \geq mp_i$ to $G_i$ instead of the optimization condition.  In the Gamut \cite{DBLP:conf/atal/NudelmanWSL04} version of the game, vendors do not choose location but prices, because there is no way to express that sellers should choose different locations in a normal form game like we do here with the {\em all\_different} constraint.
\end{example}

\begin{example}[Cloud Resource Allocation Game] \label{ex:cloud} 
Resource allocation is a central issue in cloud computing where clients use and pay computing resources on demand.  In order to manage conflicting interests between clients, \cite{JalapartiNGC:TechReport2010} has proposed the framework of {\em CRAG} (Cloud Resource Allocation Game) in which resource assignments are defined by game equilibrium. 

A cloud computing provider owns a set ${\cal M} = \{M_1, \ldots, M_m\}$ of $m$ machines, each machine $M_j$ having a capacity $c_j$ representing the amount of resource available (for example CPU-hour, memory).  The cost of using machine $j$ is given by $l_j(x) = x \times u_j$ where $x$ is the number of resources requested and $u_j$ some unit cost.  A set of $n$ clients  $\players = \{1,2,..,n\}$ wants to use simultaneously the cloud in order to perform tasks.  Client $i \in \players$ has $m_i$ tasks $\{T_{i1},...,T_{im_i}\}$ to perform, with respective requested capacity of $\{d_{i1},...,d_{im_i}\}$.  Each client $i \in \players$ chooses selfishly an allocation $r_{ik}$ for the task $T_{ik}$ ($k \in 1..m_i$) and wishes to minimize her cost $cost_i = \sum_{k = 1..m_i} l_{r_{ik}}(d_{ik})$.  We assume that the provider's resources amount is sufficient to accommodate the requests of all the clients: $\sum_{i \in [1..n]} \sum_{k \in [1..m_i]} d_{ik} \leq \sum_{j \in [1,..m]} c_j$.  This problem can be modelled by the following COG-HC:
\begin{itemize}
  \item $\players =\{1,..,n\}$
  \item $\forall i \in \players, V_i = \{r_{i1},..., r_{im_i}\}$
  \item $\forall i \in \players, \forall k \in [1,...,m_i], D(r_{ik}) = \{1, \ldots, m\}$
  \item $C$ is composed of the following constraints:
    \begin{itemize}
      \item channelling constraints for boolean variables stating that machine $j$ is requested by task $t_{ik}$: $(r_{ik} = j) \leftrightarrow (choice_{ijk} = 1)$
      \item capacity constraints: $\forall j \in [1,..,m]$, $\sum_{i \in [1..n]} \sum_{k \in [1..m_i]} choice_{ijk} \times d_{ik} \leq c_j$
    \end{itemize}
  \item  $\forall i \in \players, G_i$ is composed of the following constraint: $$cost_i = \sum_{j = 1..m} \sum_{k = 1..m_i} choice_{ijk} \times l_j(d_{ik})$$
  \item  $\forall i \in \players, Opt_i = \text{Minimize } (cost_i)$ 
\end{itemize}
\end{example}

Other interesting examples can be modeled by Constraint Games like network game \cite{DBLP:journals/networks/BouhtouEM07}, strategic scheduling \cite{Vocking:AGT2007}, or games from the Gamut suite \cite{DBLP:conf/atal/NudelmanWSL04}.

%% file: 04-pruning.tex
\section{Pruning techniques} \label{sec:pruning}
A natural algorithm is to use generate and test to find an equilibrium.  This naive algorithm is however the only known algorithm for finding PNE \cite{Turocy:PC2013} and from the complexity result, it is unlikely that any fast (polynomial) algorithm could exist.  This algorithm is therefore the basis of the implementation of the Gambit solver \cite{Gambit} for PNE enumeration.   We first show that this technique is subject to a form of trashing.  In order to simplify the exposition, we assume in the remaining of the paper that each player $i$ only controls one variable $x_i$ with domain $D_i$.  The extension to more than one variable per player is not difficult (indeed our solver Conga does not have this limitation since many examples require a player to control several variables).

The {\em enum1} algorithm (Algorithm \ref{alg:naive-solve}) consists in enumerating all strategy profiles, testing each of them for each player for deviation and skipping to the next profile when the first deviation is found.
\begin{algorithm}
  \caption{enum1 \label{alg:naive-solve}}
  \begin{algorithmic}[1]
    \Statex
    \Function{enum1}{Game $CG$}: {\bf setof} tuples
      \Let{Nash}{$\emptyset$}
      \For{$s \in D^{V_c}$}
        \If{IsNash($s$)}
          \State{Nash = Nash $\cup$ $\{s\}$}
        \EndIf
      \EndFor
      \State \Return{Nash}
    \EndFunction
    \Statex
    \Function{IsNash}{tuple $s$}: boolean
      \For{$i \in \players$}
        \For{$v \in D_i, v \neq s_i$}
          \If{$(s_{-i},v) <_{opt_i} s$}
            \State{\Return{false}}
          \EndIf
        \EndFor
      \EndFor
      \State \Return{true}
    \EndFunction
  \end{algorithmic}
\end{algorithm}

The following example shows that some deviations are performed several times.
\begin{example}
 Let $G$ be the 2-players game defined by the following table:
 \begin{center}
  \begin{tabular}{|>{\PBS{\centering}}m{5mm}@{\,}|>{\PBS{\centering}}m{5mm}@{\,}|>{\PBS{\centering}}m{15mm}@{\,}|@{\,}>{\PBS{\centering}}m{15mm}@{\,}|@{\,}>{\PBS{\centering}}m{15mm}@{\,}|} \hline
    \multicolumn{2}{|c|}{} & \multicolumn{3}{c|}{$y$} \\ \cline{3-5}
    \multicolumn{2}{|c|}{}       & 1              & 2     & 3      \\ \hline
    \multirow{3}{*}{$x$}     & a & (0,1)$_\alpha$ & (1,0) & (1,0)  \\ \cline{2-5}
                             & b & (0,1)$_\beta$  & (0,0) & (1,0)  \\ \cline{2-5}
                             & c & (1,0)          & (1,1) & (0,0)  \\ \hline
  \end{tabular}
 \end{center}
 We assume that the enumeration starts by Player $x$.  The first tuple to be enumerated is $(a,1)$ denoted by $\alpha$.  Deviation is checked for Player $y$ and no deviation is found.  Then deviation is checked for Player $x$ and a deviation towards $(c,1)$ is found.  Thus this tuple is not a PNE.  The next candidate is $(b,1)$ denoted by $\beta$.  This tuple is checked for Player $y$ and again no deviation is found.  But when checked for Player $x$, the same deviation towards $(c,1)$ is found as for tuple $\alpha$.
\end{example}

This form of trashing is a strong motivation to investigate search and pruning techniques for Constraint Games.  
To introduce our technique, we first recall \cite{DBLP:journals/jair/GottlobGS05} where the authors introduce (originally for graphical games) a CSP composed of {\em Nash constraints} to represent best responses for each player.
\begin{definition}[GGS-CSP] ~
  Let $CG=(\players,V,D,G,opt)$ be a COG.  The {\em Nash constraint} of player $i \in \players$ is $g_i = (V_c,T)$ where $T = \{t \in D^{V} ~|~ \not\!\!\exists t' \in D^{V} \mbox{ s.t. } t' <_{opt_i} t \}$.  The {\em GGS-CSP} ${\cal G}(CG)$ of $CG$ is the set of Nash constraints for all players.
\end{definition}
This CSP has the important property that it defines the PNE of the game:
\begin{theorem} (\cite{DBLP:journals/jair/GottlobGS05}) \label{th:GGS}
  $t$ is a PNE of $CG$ $\leftrightarrow$ $t \in sol({\cal G}(CG))$
\end{theorem}
Then it follows that a PNE of a Constraint Game $CG$ has a support in all of its Nash constraints.

Our technique consists to perform a traversal of the search space by assigning the variables of each player in turn according to a predefined ordering on $\players$.  For each candidate tuple, we seek for supports by performing an incremental computation of the Nash constraints.  Each computed deviation is recorded in a table for each player.  By retrieving tuples in Nash constraints, we can avoid computing costly deviations.

However, since we are studying general games, each Nash constraint has the same arity as the whole problem, which is challenging in terms of space requirements.  First, note that any tuple deleted from a table does not hinder the correctness of the Nash test.  It may only forces a deviation to be computed twice.  Hence we are free to limit the size of the tables and trade space with time.  In practice, two features limit the size of the tables.

First, deviation checks are performed in reverse player ordering.  It means that a tuple checked for the first player must have succeed the deviation test for all other players.  In practice for most problems, this limits the number of tuples reaching the upper levels.  Second, we can delete a tuple $t$ recorded in a table when we can ensure that no candidate $t'$ will deviate anymore towards $t$.  This property is given by the following independence of subgames theorem.  Let $CG=(\players,V,D,G,opt)$ be a constraint game.  A game $CG'=(\players,V,D',G,opt)$ is a subgame of $CG$ if $\exists i \in \players, D'_i \subseteq D_i$ and $i \neq j \rightarrow D'_j=D_j$.  We denote by $br_i(t) = \{t' \in sol(G_i) ~|~ t'_{-i} = t_{-i} \}$ the set of best response strategies from $t$ for Player $i$.  
\begin{figure}[h]
 \centering
 \includegraphics[height=3cm]{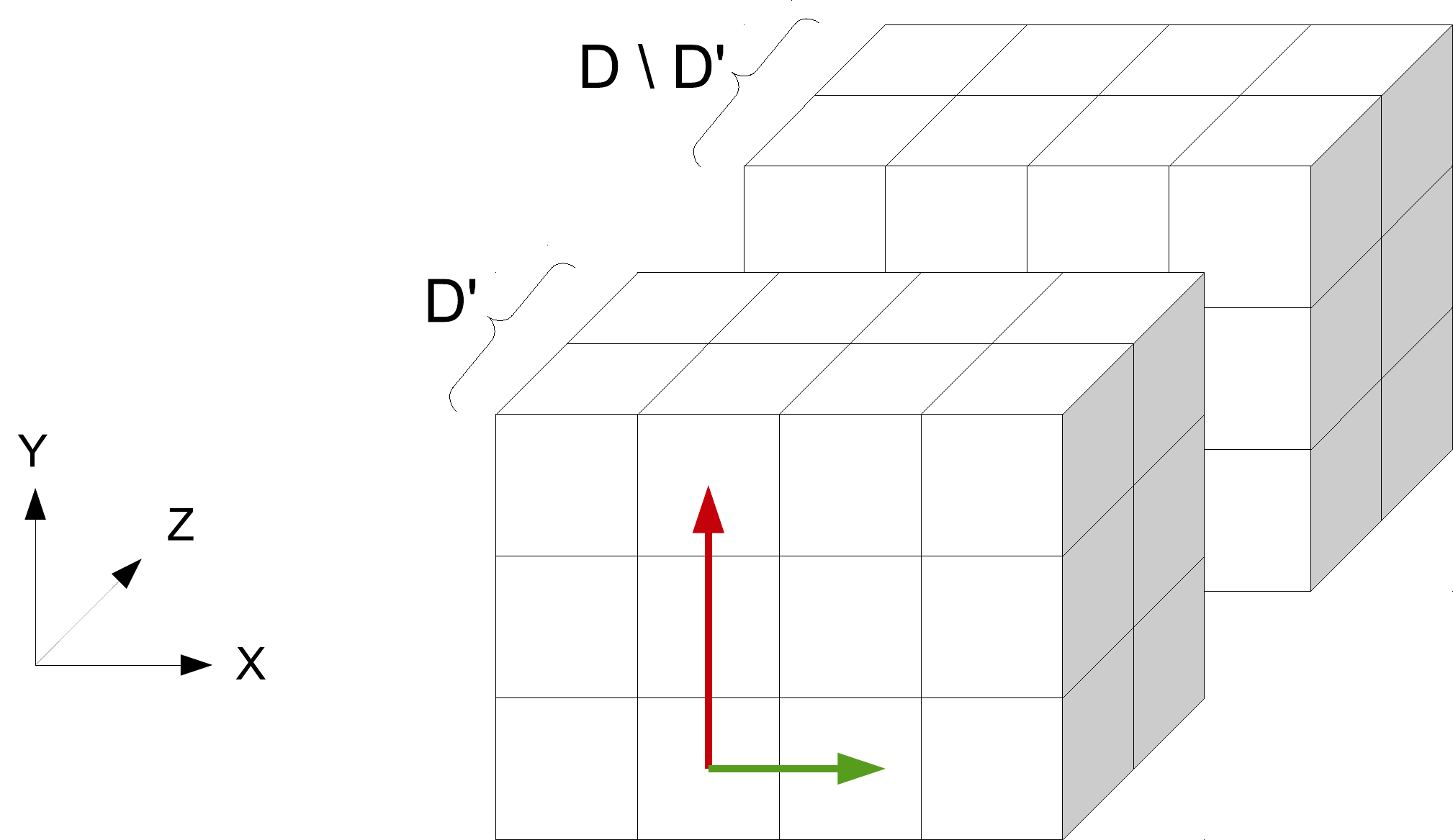}
 \caption{Independence of subgames}
 \label{fig:theorem1}
\end{figure}
\begin{proposition} \label{prop:ios}
 Let $CG$ be a constraint game and $CG'$ a subgame of $CG$ such that $D'_i \subseteq D_i$.  Let $D''=D \setminus D'$, $t' \in D'^V$ and $t'' \in D''^V$.  Then $\forall j \in \players, j \neq i \rightarrow br_j(t') \cap br_j(t'') = \emptyset$.
\end{proposition}
\begin{proof}
  The proof is by contradiction.  Suppose there exists $t' \in D'^V$ and $t'' \in D''^V$ such that $br_j(t') \cap br_j(t'') \neq \emptyset$.  Let $s \in br_j(t') \cap br_j(t'')$.  Then since $s \in br_j(t')$, $s_i \in D'^{V_i}$.  Since $s \in br_j(t'')$, $s_i \in D''^{V_i}$.  Hence the contradiction. $\Box$
\end{proof}
This proposition is illustrated in Figure \ref{fig:theorem1}: if we split the search space following Player $Z$, best responses for Players $X$ and $Y$ are forced to remain in different subspaces.

By applying inductively Proposition \ref{prop:ios} on a sequence of assignments of strategies for Player $1$ to $k$ with $k<n$, we see that two branches of the search tree will not share best responses for the remaining unassigned players.  Hence we can freely remove all tuples from the table of subsequent players once the branch is explored.

The last optimization consists in elimination of {\em never best responses} (NBR).
\begin{definition}
 A strategy $s_i$ for Player $i$ is a {\em never best response} if $\forall t_{-i}, \exists s'_i$ such that $(s'_i,t_{-i}) <_{opt_i} (s_i,t_{-i})$.
\end{definition}
Iterative elimination of NBR is a sound pruning for games \cite{DBLP:conf/tark/Apt05} which additionally is stronger than elimination of strongly dominated strategies.  But unfortunately their detection is very costly in the $n$-players case since it needs to know that this action will never been chosen by Player $i$ for all strategy profiles of the other players.  However, being a NBR in a subgame is a sufficient condition for not being an equilibrium:
\begin{proposition} \label{prop:inbr}
 Let $CG$ be a constraint game and $CG'$ a subgame of $CG$ such that $D'_i \subseteq D_i$.  Let $s_j \in D'_j$ be a NBR in $CG'$ with $j \neq i$.  Then for all $s_{-j}$, if $s=(s_j,s_{-j})$ is a PNE, then $s_i \not\in D'_i$.
\end{proposition}
\begin{proof}
  The proof is by contradiction.  Suppose there exists a PNE $s=(s_j,s_{-j})$ with $(s_{-j})_i \in D'_i$.  Because $s$ is a PNE, we have $\forall k \in \players, s_k \in br(s)|_k$.  Then because $s_j$ is a NBR for $j$ in $CG'$, there exists $s'_j$ such that $(s'_j,s_{-j}) <_{opt_j} (s_j,s_{-j})$.  Thus $s_j \not\in br(s)|_j$.  $\Box$
\end{proof}
By applying inductively Proposition \ref{prop:inbr} on a sequence of assignments of strategies for Player $1$ to $k$ with $k<n$, we see that if we detect that a value $v$ is NBR for player $k$ in the subgame defined by the sequence of assignments, then this value will not participate to a PNE and we can prune it.

%% file: 05-nash-enumeration.tex
\section{An algorithm for Nash equilibria enumeration} \label{sec:enum}

\setlength{\textfloatsep}{10pt plus 1.0pt minus 2.0pt}
\begin{algorithm}[H]
  \caption{ConGa \label{alg:cg-solve}}
  \begin{algorithmic}[1]
    \Statex
    \State{\bf global:}
    \State{~~~~~~$BR$: array[1..n] of tuples}    \Comment{\em best responses for all players}
    \State{~~~~~~$cnt$: array[1..n] of integer}  \Comment{\em counters for NBR detection}
    \State{~~~~~~$Nash$: set of tuples}          \Comment{\em Nash equilibria}
    \State{~~~~~~$S$: global solver} 
    \Statex
    \Function{ConGa}{Game $CG$}: {\bf setof} tuples
      \Let{Nash}{$\emptyset$}
      \State{Initialize solver $S$ with hard constraints}
      \Let{$A$}{$D$}
      \State{enum($A,1$)}
      \State \Return{Nash}
    \EndFunction
  \end{algorithmic}
\end{algorithm}

We propose a tree-search algorithm for finding all PNE.  The general method is based on three ideas:
\begin{itemize}
  \item all candidates (except those which are detected as NBR) are generated in lexicographic order;
  \item undominated solutions for each player are recorded in a table;
  \item whenever a domination check is performed, it first checks this player's recorded best responses.
\end{itemize}

We assume given an ordering on players from $1$ to $n$.  The main algorithm (Algorithm \ref{alg:cg-solve}) launches the recursive traversal (Algorithm \ref{alg:enum}) starting by Player $1$.  We distinguish the original domains of the variables (called $D$) used to compute deviations from their actual domain explored by the search tree (called $A$) and subject to pruning by arc-consistency on hard constraints.

Propagation of hard constraints allows to ensure that no forbidden tuple will be explored.  If the propagation returns {\em false}, then at least one domain has been wiped out and there is no solution in this subspace.  Otherwise domains $A$ are reduced according to arc-consistency.  Values for each player are then recursively enumerated.  When a tuple is reached, it is checked for Nash condition (line 5) by Algorithm \ref{alg:checkNash}.  Otherwise, at least one domain remains to be explored.  Each player $i$ owns a table $BR[i]$ of best responses, initialized empty and a counter $cnt[i]$ initialized with the size of the subspace needed to detect potential never best responses.  For the sake of efficiency, the table $BR[i]$ is actually implemented by a search tree.  Hence insertion and search are done in $O(|\players|)$.  After the recursive call of {\em enum}, we test whether all the subspace after Player $i$ has been checked for deviation.  Then all subsequent values are {\em never best responses}.  In this case an exit from the loop causes backjumping to the ancestor node.  This backjumping is done after the exploration by {\em checkEndOfTable} (Algorithm \ref{alg:checktable}) of the potential values which are stored in the table and belong to the unexplored space (lines 14-17).

\begin{algorithm}[H]
  \caption{enum \label{alg:enum}}
  \begin{algorithmic}[1]
    \Statex
    \Procedure{enum}{domains $A$, int $i$}
      \Let{status}{$S$.propagate($A$)}
      \If{status}
        \If{$i > n$}
          \State{checkNash(tuple($A$),$n$)}
        \Else
          \Let{BR[$i$]}{$\emptyset$}
          \Let{cnt[$i$]}{$\Pi_{j>i} |D_j|$}
          \While{$A_i \neq \emptyset$}
            \State{choose $v \in A_i$}
            \Let{$B$}{$A$}
            \State{enum($(B_{-i},(B_i=\{v\})),i+1$)}
            \Let{$A_i$}{$A_i - \{v\}$}
            \If{$cnt$[$i$] $\leq 0$}
              \State{checkEndOfTable($A,i$)}
              \State{break}
            \EndIf
          \EndWhile
        \EndIf
      \EndIf
    \EndProcedure
  \end{algorithmic}
\end{algorithm}

The {\em checkNash} procedure in Algorithm \ref{alg:checkNash} verify whether a player can make a beneficial deviation from a tuple.  Since the exploration of the search tree is done level by level, the verification starts from the deepest level.  First the tuple is searched in the table of stored best response for this player (line 5).  If not found, a solver for $G_i$ is called in function {\em deviation} depicted in Algorithm \ref{alg:deviation} (line 7).  This function returns the set $d$ of deviations for Player $i$ from a tuple $t$.  There can be more than one deviation.  In a CSG, it means that several assignments satisfy the constraints of $G_i$.  In a COG, it means that the optimal value is reached for more than one point.  If $d$ is empty, it means that there is no possible action for Player $i$ which can satisfy the constraints of her goal.  Indeed, a tuple can be an equilibrium even if a (or all) player is unsatisfied.  In this case we return the whole initial domain as deviation: any value can participate to a PNE because the player has no preference.

\begin{algorithm}[H]
  \caption{checkNash \label{alg:checkNash}}
  \small
  \begin{algorithmic}[1]
    \Statex
    \Procedure{checkNash}{tuple $t$, int $i$}
      \If{$i = 0$}
        \Let{$Nash$}{$Nash \cup \{ t \}$}
      \Else
        \Let{d}{search\_table($t,BR,i$)}
        \If{$d = \emptyset$}
          \Let{$d$}{deviation($t,i$)}
          \If{$d = \emptyset$}
            \Let{$d$}{$D_i$}
          \EndIf
          \State{insert\_table($i,BR,d$)}
          \State{$cnt$[$i$] -\,-}
        \EndIf
        \If{$t_i \in d$}
          \State{checkNash($t,i-1$)}
        \EndIf
      \EndIf
    \EndProcedure
  \end{algorithmic}
\end{algorithm}

\begin{algorithm}
  \caption{deviation \label{alg:deviation}}
  \begin{algorithmic}[1]
    \Statex
    \Function{deviation}{tuple $t$, int $i$}: set of {\bf integer}
      \Let{$d$}{$\emptyset$}
      \State{Initialize solver $S_i$ with $G_i$ (and $opt_i$ for a COG)}
      \State{add constraints $x_j = t_j$ for all $j \neq i$}
      \Let{$sol$}{$S_i$.getSolution()}
      \While{$sol \neq nil$}
        \Let{$d$}{$d \cup \{sol\}$}
        \Let{$sol$}{$S_i$.getSolution()}
      \EndWhile
      \State{\Return{$d$}}
    \EndFunction
  \end{algorithmic}
\end{algorithm}

\begin{algorithm}
  \caption{checkEndOfTable \label{alg:checktable}}
  \begin{algorithmic}[1]
    \Statex
    \Procedure{checkEndOfTable}{domain $A$, int $i$}
      \ForAll{$t \in$ $BR$[$i$] such that $t \in \Pi_{i=1..n} A_i$}
        \State{checkNash($t, n$)}
      \EndFor
    \EndProcedure
  \end{algorithmic}
\end{algorithm}

The procedure {\em checkEndOfTable} depicted in Algorithm \ref{alg:checktable} is used when the subspace has been explored and just before performing backjumping.  In this case, all tuples of the table which belong to the unexplored zone are checked for PNE.  An example of backjumping is given in Figure \ref{fig:nbr}.

\begin{figure}[H]
 \centering
 \includegraphics[height=4cm]{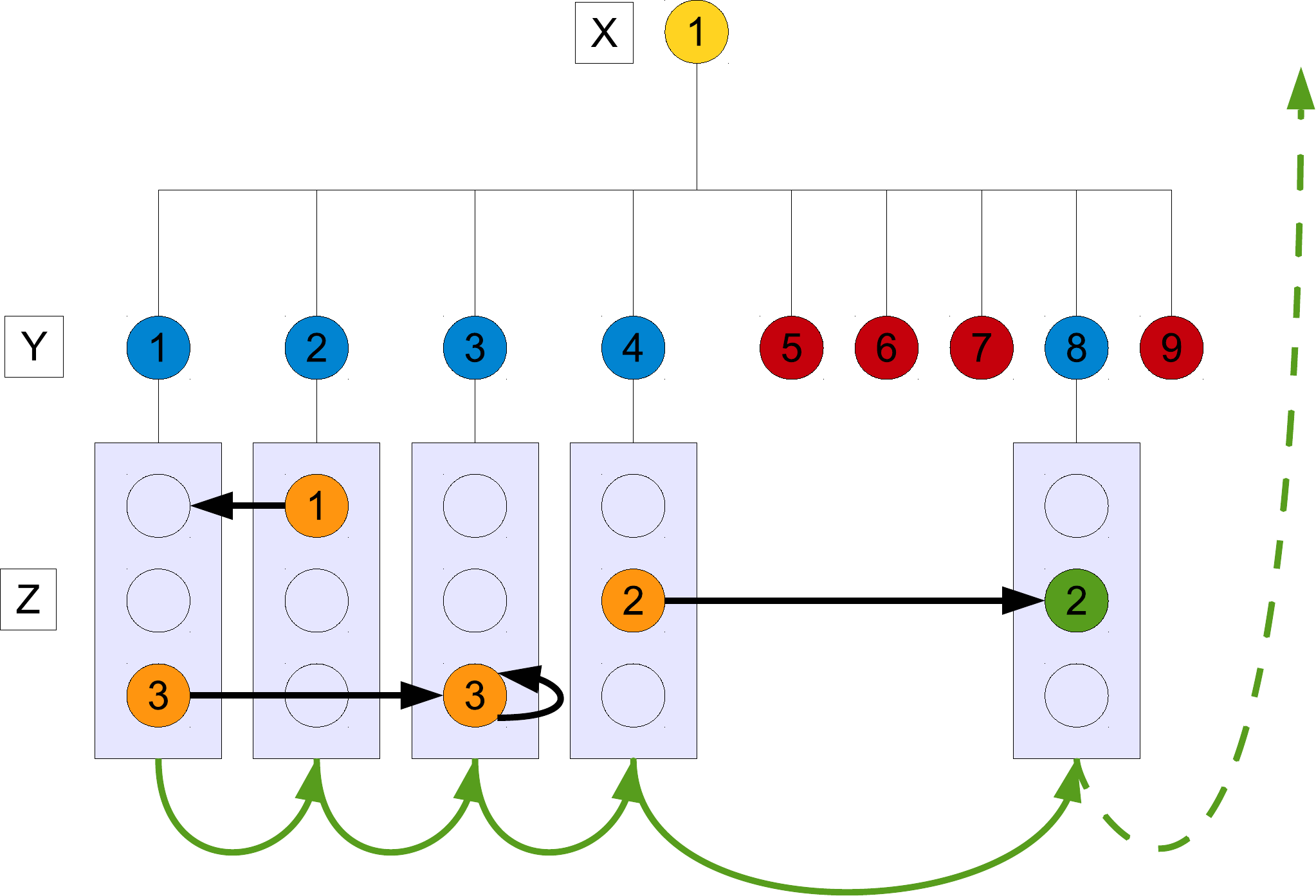}
 \caption{Online detection of never best responses}
 \label{fig:nbr}
\end{figure}
In this example, the domains of player $Y$ and $Z$ are respectively of size 9 and 3.  Hence $cnt[Y]$ is initialized to $3$.  All tested tuples are of the form $1yz$ where $1$ is the value on $X$.  If we suppose that tuples $113$, $121$, $133$ and $142$ are stable for $Z$, these tuples are lifted to $Y$ level to be checked for $Y$'s deviation.  Solid arrows depict the deviations recorded for Player $Y$.  It happens in this example that only by exploring values $1$, $2$, $3$ and $4$ for $Y$ yield a complete traversal of the subspace defined by $Z$ (all values from $Z$'s domain have been considered).  Thus after the exploration of $Y=4$, we know that only $182$ recorded in $Y$'s table can be a PNE with $X=1$.  The other values of Player $Y$ are NBR.  It is sufficient to check this tuple by {\em checkEndOfTable} and we can backtrack to the next value of Player $X$ (dotted arrows).  In general {\em checkEndOfTable} tests all tuples of $BR[Y]$ which belong to the unexplored part of the search space.  This NBR detection is incomplete but comes almost for free because it only takes a counter.  Note that by Proposition \ref{prop:ios}, when exploring $X=2$, the table for $Y$ can be reset because no other tuple will deviate to a tuple where $X=1$.

We propose to follow the Conga algorithm on the small example depicted on Figure \ref{fig:conga}.
\begin{figure}[h]
 \centerline{\includegraphics[width=12cm]{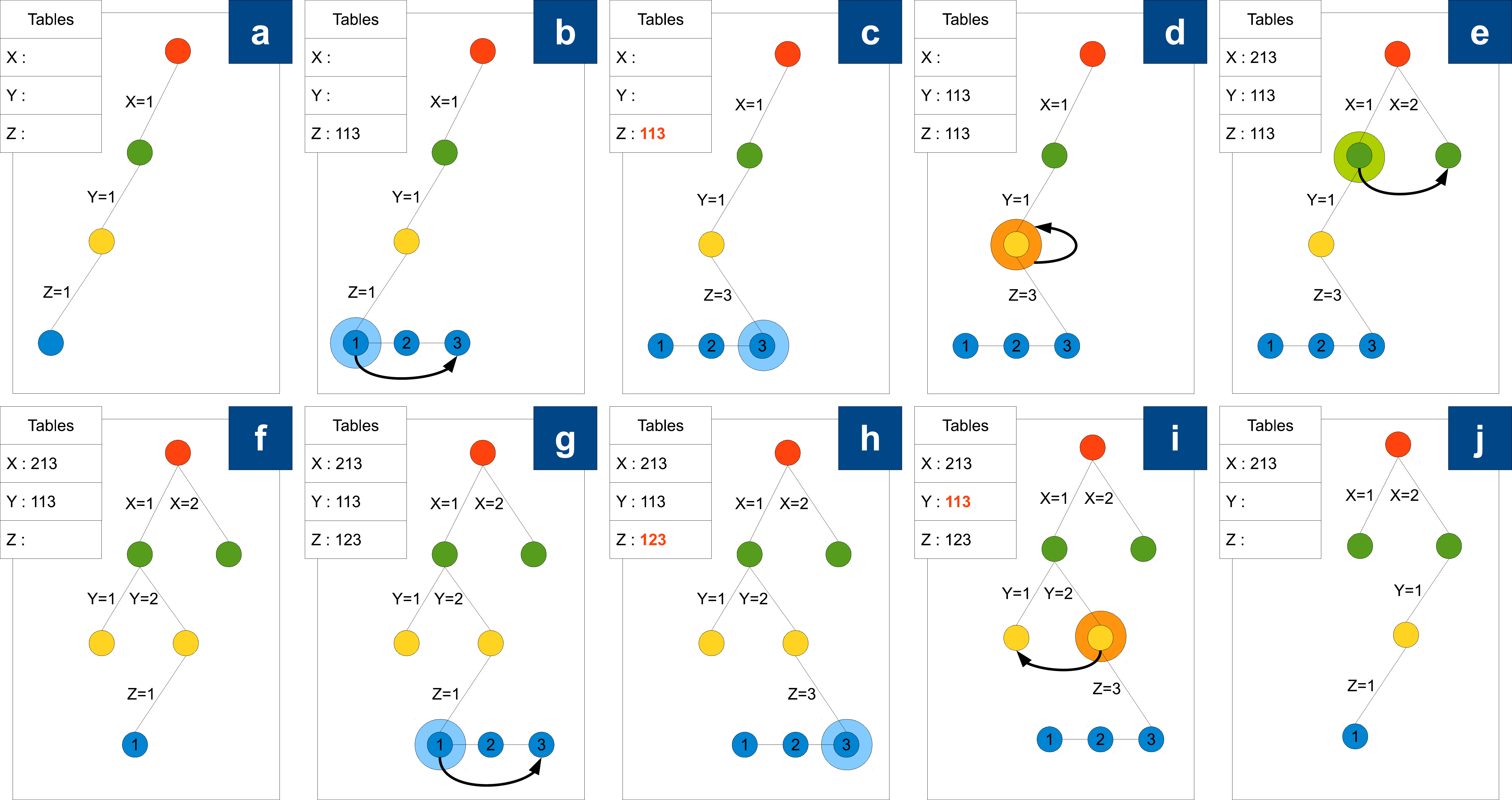}}
 \caption{Conga algorithm}
 \label{fig:conga}
\end{figure}

The resolution starts in Figure \ref{fig:conga}.{\bf a} (abbreviated \ref{fig:conga}.{\bf a}) with tuple $111$.  This tuple is checked for deviation for $Z$ and deviates, say, to $Z=3$ as shown in \ref{fig:conga}.{\bf b}, storing $113$ in Z's table.  Since the possible deviations for $Z$ are explored, we jump by {\em checkEndOfTable} to $113$ as in \ref{fig:conga}.{\bf c}.  This tuple is obviously a Nash candidate for $Z$, and is checked for deviation for the other players in reverse order.  It is thus checked for $Y$ in \ref{fig:conga}.{\bf d} and, in our example, is found stable for $Y$, storing $113$ in $Y$'s table.  Then the tuple is checked for $X$ in \ref{fig:conga}.{\bf e} and a deviation to $213$ is found, which causes $213$ to be stored in $X$'s table.  Then backtracking occurs in \ref{fig:conga}.{\bf f} at $Y$'s level, which resets the table for $Z$.  The next candidate is $121$.  In \ref{fig:conga}.{\bf g}, a deviation is found for $Z$ to $123$, in \ref{fig:conga}.{\bf h}, the tuple $123$ is considered, and checked for deviation for $Y$ in \ref{fig:conga}.{\bf i}.  But it is no use to perform a real deviation check because the value $113$ is found in $Y$'s table, meaning that a best response with $Y=1$ has been recorded for $X=1,Z=3$ and is thus a deviation for $123$.  Backtracking at $X$'s level occurs and the tables for $Z$ and $Y$ are emptied (\ref{fig:conga}.{\bf j}).  Note that the tables for Player $i$ are actually implemented by a tree whose nodes represent all players but $i$, with $i$'s best responses attached on the leaves.  Thus the search for deviation in the table does not depend on the number of recorded tuples but only on the number of players and is performed in $O(|\players|)$.

\begin{proposition}
 ConGa is correct and complete.
\end{proposition}
\begin{proof}
 Correctness comes from the correctness of PNE check (theorem \ref{th:GGS}).  A reported PNE has been checked for deviation for every player.  Either the tuple has been recorded in the table as deviation from another one, or had been directly checked by the solver against the player's goal.  Completeness is due to the traversal of the whole search space and soundness of never best response pruning. $\Box$
\end{proof}

%% file: 06-experiments.tex
\section{Experiments} \label{sec:xp}
\vspace*{-1mm}
We have performed experiments on classical games of the Gamut suite \cite{DBLP:conf/atal/NudelmanWSL04} and some games with hard constraints.  Results are summarized in Table \ref{tab:bench} in which the name of the game is followed by the number of players and the size of the domain. Gamut games are CG (Congestion Game), GTTA (Guess Two Third Average), LG(GV) (Location Game, Gamut version), MEG (Minimum Effort Game) and TD (Traveller's Dilemma).  Their description can be found in \cite{DBLP:conf/atal/NudelmanWSL04}.  The other games are LG(HC) (Location Game with Hard Constraints, example \ref{ex:location}) and CRAG (Cloud Resource Allocation Game, example \ref{ex:cloud}).

\begin{table}[h]
\vspace*{-2mm}
 \centering
 {\small
 \begin{tabular}{|l|r|r|r|r|r|r|r|r|r|r|}\hline
   \textbf{Name}  & \multicolumn{2}{c|}{\textbf{NF gen}} & \textbf{Gambit} & \multicolumn{3}{c|}{\textbf{enum1}} & \multicolumn{3}{c|}{\textbf{ConGa}} & \textbf{\#PNE}  \\\hline
 & {\bf \em Time}   & {\bf \em Size} & {\bf \em Time}  & {\bf \em Time} & {\bf \em \#Cand}  & {\bf \em \#Dev} & {\bf \em Time} &  {\bf \em \#Cand}  & {\bf \em \#Dev} &   \\
  \Xhline{2\arrayrulewidth}
   
   CG.7.15        & 253  & 5.1      & {\em MO}        & 70   & 1.7E+8 & 1.8E+8    & {\bf 27}  & 2.1E+7 & 1.3E+7       & 630  \\\hline
   CG.8.15        & 4613 & 89       & {\em MO}        & 1019 & 2.5E+9 & 2.7E+9    & {\bf 371} & 3.1E+8 & 1.9E+8       & 1680 \\\hline
   CG.9.15        & {\em TO} & --   & --              & 17361& 3.8E+10 & 4.2E+10  & {\bf 5880}& 4.9E+9 & 2.9E+9       & 5040 \\
   \Xhline{2\arrayrulewidth}

   GTTA.3.100     & 1   & 0.1       & 17              & 4   & 1.0E+6 & 1.3E+6     & {\bf 0}  & 1.0E+4 & 1.0E+4       & 1    \\\hline
   GTTA.4.100     & 113 & 1.7       & 1844            & 312 & 1.0E+8 & 1.3E+8     & {\bf 10} & 1.0E+6 & 1.0E+6       & 1    \\\hline
   GTTA.5.100     & {\em TO} & 205  & {\em MO}        & 4032& 1.0E+10& 1.3E+10    & {\bf 778}& 1.0E+8 & 1.0E+8       & 1    \\
   
  \Xhline{2\arrayrulewidth}   
   LG(GV).2.1000  & 1 & 0.01       & 134             & 339  & 1.0E+6 & 1.0E+6    & {\bf 6}   & 2.0E+3 & 1.5E+3       & 0    \\\hline
   LG(GV).2.2000  & 6 & 0.04       & 655             & 1441 & 4.0E+6 & 4.0E+6    & {\bf 31}  & 4.0E+3 & 3.5E+3       & 0    \\\hline
   LG(GV).2.3500  & 17 & 0.1       & 5337            & 6789 & 1.2E+7 & 1.2E+7    & {\bf 93}  & 7.0E+3 & 6.0E+3       & 0    \\\hline
   LG(GV).2.5000  & 34 & 0.2       & 7786            & 20000& 2.5E+7 & 2.5E+7    & {\bf 201} & 1.0E+4 & 9.0E+3       & 0    \\\hline
   LG(GV).2.20000 & 552& 3.7       & {\em MO}        & {\em TO}  & -- & --       & {\bf 3578}& 4.0E+4 & 3.9E+5       & 0    \\
  \Xhline{2\arrayrulewidth}      
   MEG.3.100      & 1 & 0.1        & 13              & 0    &1.0E+6 & 1.0E+6     & {\bf 0}   & 1.9E+4 & 1.5E+4       & 100  \\\hline
   MEG.4.100      & 91& 1.9        & 1555            & 28   &1.0E+8 & 1.0E+8     & {\bf 6}   & 1.9E+6 & 1.3E+6       & 100  \\\hline
   MEG.5.100      & {\em TO}& 241  & {\em MO}        & 2082 &1.0E+10 & 1.0E+10   & {\bf 403} & 1.9E+8 & 1.2E+8       & 100  \\\hline
   
   MEG.30.2       & 8784& 91       & {\em MO}        & {\bf 423}&1.1E+9&2.1E+9   & 503       & 5.4E+8 & 1.1E+9       & 2    \\\hline
   MEG.35.2       & TO & --        & --              & {\bf 15933}&3.4E+10&6.9E+10& {\em TO} & --     & --           & 2    \\
    \Xhline{2\arrayrulewidth}   
   TD.3.99        & 3  & 0.1       & 14              & 0  & 9.7E+5 & 9.8E+5      & {\bf 0} & 1.9E+4   & 1.5E+4       & 1    \\\hline
   TD.4.99        & 76 & 1.9       & 1572            & 26 & 9.6E+7 & 9.7E+7      & {\bf 7} & 1.9E+6   & 1.3E+6       & 1    \\\hline
   TD.5.99        & 8930 & 119     & {\em MO}        & 2028 & 9.1E+9 & 9.6E+9    & {\bf 446} & 1.8E+8 & 1.2E+8       & 1    \\
 
   \Xhline{2\arrayrulewidth}
  
   CRAG.7.9       & N/A  & N/A      & N/A             & 323  & 4.7E+6 & 5.3E+6    & {\bf 57}  & 1.0E+6 & 5.9E+5 	  & 1    \\\hline
   CRAG.8.9       & N/A  & N/A      & N/A             & 3300 & 4.2E+7 & 4.8E+7    & {\bf 540} & 9.5E+6 & 5.3E+6       & 1    \\\hline
   CRAG.9.9       & N/A  & N/A      & N/A             & -- & 3.8E+8 & 4.3E+8      & {\bf 5022}& 4.3E+7 & 4.8E+7       & 1    \\
   \Xhline{2\arrayrulewidth}
   LG(HC).4.30    & N/A  & N/A      & N/A             & 26   & 6.5E+5 & 8.0E+5    & {\bf 6} & 1.4E+5 & 4.4E+4         & 24   \\\hline
   LG(HC).5.30    & N/A  & N/A      & N/A             & 778  & 1.7E+7 & 2.1E+7    & {\bf 257}  & 4.1E+6 & 1.2E+5      & 240  \\\hline
   LG(HC).6.30    & N/A  & N/A      & N/A             & {\em TO} & -- & --        & {\bf 13180} & 1.1E+8 & 3.2E+7     & 2160 \\
   \Xhline{2\arrayrulewidth}
   
 \end{tabular}
 }
 \caption{Results for Gamut and other games}
 \label{tab:bench}
\end{table}

For each instance, we have compared ConGa to the game solver Gambit \cite{Gambit} and to a base solver called {\em enum1} (Algorithm \ref{alg:naive-solve}).  This solver works like Gambit by examining each tuple but the only difference is that it uses the compact Constraint Game representation.  All experiments have been executed on a 48-cores AMD Opteron 6174 with 4-processors at 2,2 GHz with 256 GB of RAM.  In table \ref{tab:bench}, times are given in seconds and the number {\em aE+b} is equal to $a \times 10^b$.  

The experiment on Gambit is divided into two steps: we first generate the normal form matrix (column NF gen) and then we launch the solver {\em gambit-enumpure} on the normal form to find all PNE.  We have measured the time needed to generate the normal form (with a time-out of 9 000 seconds) and its size (in GB), then the time required for the game to be solved by Gambit (with a time-out of 20 000 seconds).  TO stands for {\em Time Out}, MO for {\em Memory Out} and "--" means there is no information (for example, if the generation times out, it is not possible to launch the resolution).  As expected, the size of the normal form soon becomes intractable and exceeds the capacity of Gambit, although it can handle matrices of 2 GB.

For {\em enum1} and {\em ConGa}, we have measured the time needed to solve an instance (with also a time-out of 20 000 seconds), the number of candidate profiles and the number of deviation checks performed when checking whether a candidate is a PNE.  From a simple reasoning, the number of candidate for {\em enum1} is simply $|D^{V_c}|$, and the number of checks is comprised between the number of candidates and an upper bound of $|D^{V_c}| \times |\players|$.  Not surprisingly, we see that {\em ConGa} prunes most of the time a large part of the search space, mainly thanks to NBR detection.  But interestingly, most of the time, it also saves deviation checks, meaning that the solution is found in the tables before a check is performed.  A notable counterexample is the Minimum Effort Game with a domain of size 2 (MEG.30.2 and MEG.35.2) for which neither the tables or the NBR detection are working because the domains are too small.  Note that games with hard constraints are not implementable in Gambit, indicated by N/A for {\em not applicable}.  In all other benchmarks, Conga outperforms both Gambit and {\em enum1} by one order of magnitude or even more.

A potential problem of Conga could be that the size of the tables grows too much.  It is easy to build an example for which the first player will get a table of exponential size: the game with no constraint for each player.  In this game  each profile is a PNE and is thus stored in the table of the first player.  However, this behavior has not been observed in practice in any of our examples.  Tables rather stay of reasonable size, either because they belong to lower level players and they are reset often or because many profiles do not reach high level players.

\vspace*{-3mm}

%% file: 07-conclusion.tex
\section{Conclusion} \label{sec:conclusion}
Probably the most important feature of Constraint Games is that they provide a new compact yet natural encoding to game models.  In this paper, we propose the first complete solver for constraint games based on a fast computation of Nash consistency and pruning of Never Best Responses.  We show that these techniques are yet able to outperform the existing state-of-the-art solver Gambit.  But they also open directions for further investigations on efficient algorithmic techniques to compute Nash equilibria.   Future work include the use of heuristics, graphical constraint games and other solution concepts like Pareto equilibria.
